\documentclass[11pt]{article}

\usepackage{epsf,epsfig,graphicx}

\usepackage{mathrsfs}
\usepackage{amsmath,amssymb,amsthm}
\usepackage{enumerate}
\usepackage{stackrel}

 \newtheorem{theorem}{Theorem}[section]
 \newtheorem{lemma}{Lemma}[section]
 
 \newtheorem{definition}{Definition}[section]

\theoremstyle{definition}

\theoremstyle{plain}
\newcounter{algleo}[theorem]
\newlength{\lefttab}
\newlength{\numberoffset}
  {\trivlist
   \topsep=0pt\parsep=0pt\itemsep=0pt
   \addtolength{\lefttab}{1.25em}
   \addtolength{\numberoffset}{1.25em}
   \leftskip=\lefttab}%
  {\endtrivlist}

\begin{document}

\title{Minimum Label $s$-$t$ Cut has Large Integrality Gaps
\thanks{This paper is the full version of part of results presented in
the conference paper (extended abstract) \cite{TZ12} appeared in the Proceedings 
of the 10th Latin American Theoretical Informatics Symposium (LATIN).}}

\author{
Peng Zhang
  \thanks{Corresponding author. School of Software
    and School of Computer Science and Technology,
    Shandong University, Jinan, 250101, China.
    E-mail: {\tt algzhang@sdu.edu.cn}.}
\and
Linqing Tang
  \thanks{State Key Lab. of Computer Science, Institute of Software,
    Chinese Academy of Sciences, Beijing, 100190, China.
    E-mail: {\tt linqing@ios.ac.cn}.}
}

\maketitle

\begin{abstract}
Given a graph $G=(V,E)$ with a label set $L = \{\ell_1, \ell_2, \ldots,
\ell_q \}$, in which each edge has a label from $L$, a source $s \in V$,
and a sink $t \in V$, the {\sf Min Label $s$-$t$ Cut} problem
asks to pick a set $L' \subseteq L$ of labels with minimized cardinality,
such that the removal of all edges with labels in $L'$ from $G$ disconnects
$s$ and $t$. This problem comes from many applications in real world,
for example, information security and computer networks. In this paper,
we study two linear programs for {\sf Min Label $s$-$t$ Cut},
proving that both of them have large integrality gaps, namely,
$\Omega(m)$ and $\Omega(m^{1/3-\epsilon})$ for the respective linear programs,
where $m$ is the number of edges in the graph and $\epsilon > 0$
is any arbitrarily small constant.
As {\sf Min Label $s$-$t$ Cut} is NP-hard and the linear programming
technique is a main approach to design approximation algorithms,
our results give negative answer to the hope that designs better approximation
algorithms for {\sf Min Label $s$-$t$ Cut} that purely rely on linear
programming.
\end{abstract}

%\begin{keyword}
%%% keywords here, in the form: keyword \sep keyword
%Label Cut \sep Minimum Cut \sep Linear Programming \sep Integrality Gap
%\sep Approximation Algorithm
%
%%% MSC codes here, in the form: \MSC code \sep code
%%% or \MSC[2008] code \sep code (2000 is the default)
%\MSC 68W25 \sep 90C27
%%% 68W25 -- approximation algorithms
%%% 90C27 -- combinatorial optimization
%\end{keyword}

\section{Introduction}
The {\sf Min Label $s$-$t$ Cut} problem ({\sf Label $s$-$t$ Cut} for short)
is a fundamental problem in combinatorial optimization
which attracts much attention of researchers recently.

\begin{definition}
{\bf The {\sf Min Label $s$-$t$ Cut} problem.}

{\em Instance:} We are given a (directed or undirected) graph $G = (V, E)$,
a source $s \in V$, a sink $t \in V$, and a label set
$L = \{\ell_1, \ell_2, \cdots, \ell_q\}$. Each edge in graph $G$ has
a label from $L$.

{\em Goal:} A label subset $L' \subseteq L$ is called a {\em label $s$-$t$ cut},
if the removal of all edges with labels in $L'$ from $G$ disconnects
$s$ and $t$ (that is, disconnects all $s$-$t$ paths).
The goal of the problem is to find a minimum size label $s$-$t$ cut.
\end{definition}

The {\sf Label $s$-$t$ Cut} problem is quiet natural that it may appear in
many applications. For example, the researchers independently got this problem
from the study of system security \cite{JSW02,SHJ+02,SW04} and
and the study of computer networks \cite{CDP+07}. For completeness,
we give a brief introduction to the origins of the {\sf Label $s$-$t$ Cut} problem
in Appendix \ref{sec - origins of label s-t cut}.

The {\sf Min $s$-$t$ Cut} is one of the most fundamental problems in operations
research and computer science. Given a (directed or undirected) graph
and a vertex pair $(s, t)$, the problem asks to find an edge set with the minimum
size such that the removal of these edges from $G$ disconnects $s$ and $t$.
It can be easily seen that the {\sf Label $s$-$t$ Cut} problem
is in fact an edge-classified {\sf Min $s$-$t$ Cut} problem,
in which the edges are classified into groups (or types) according to their labels.
In the {\sf Label $s$-$t$ Cut} problem, we can remove a group of edges with the same
label by only paying a unit cost.
In other words, the {\sf Label $s$-$t$ Cut} problem is a natural generalization
of the classic {\sf Min $s$-$t$ Cut} problem, in the sense that
{\sf Min $s$-$t$ Cut} can be viewed as a special case of {\sf Label $s$-$t$ Cut}
in which each edge has a unique label. It is well-known that {\sf Min $s$-$t$ Cut}
can be solved in polynomial time (see, e.g., \cite[Chapter 7]{AMO93}).
However, {\sf Label $s$-$t$ Cut} is NP-hard and has very high approximation
hardness (see the related work in Section \ref{sec - related work}).

Besides the {\sf Label $s$-$t$ Cut} problem, there are still many classic
optimization problems that have been considered under the edge-classified model,
such as the {\sf Min Label Spanning Tree} problem \cite{CL97,KW98},
the {\sf Min Label $s$-$t$ Path} problem \cite{BLWZ05,HMS07},
the {\sf Min Label Traveling Salesman} problem \cite{CGMT10,XGW07},
the {\sf Min Label Perfect Matching} problem \cite{Mon05},
and the {\sf Min Label Steiner Tree} problem \cite{CMDM10}, etc.

\subsection{Related Work}
\label{sec - related work}
%{\bf Related work for {\sf Label $s$-$t$ Cut}.}
Jha et al. \cite{JSW02} proved that {\sf Label $s$-$t$ Cut} is NP-hard by reducing
the {\sf Hitting Set} problem to it. Coudert et al. \cite{CDP+07} proved that
the {\sf Label $s$-$t$ Cut} problem is NP-hard and APX-hard by reducing the {\sf MAX 3SAT}
problem to it. Zhang et al. \cite{ZCTZ11} gave the first non-trivial approximation
algorithm for the {\sf Label $s$-$t$ Cut} problem in general graphs
with approximation ratio $O(m^{1/2})$, where $m$ is the number of edges
in graph $G$.

Using a mixed strategy of LP-rounding and (any exact algorithm for) min cut,
In 2012, Tang et al. \cite{TZ12} gave
an $O(\frac{m^{1/2}}{OPT^{1/2}})$-approximation and
an $O(\frac{n^{2/3}}{OPT^{1/3}})$-approximation for {\sf Label $s$-$t$ Cut},
where $m$ is the edge number, $n$ is the vertex number, and $OPT$
is the optimal value. Note that $m$ would be $\Omega(n^2)$ in dense graphs,
implying that the two ratios $O(\frac{m^{1/2}}{OPT^{1/2}})$ and
$O(\frac{n^{2/3}}{OPT^{1/3}})$ are incomparable.
To the best of our knowledge, they are the best known
approximation ratios (in terms of $m$ and $n$, respectively)
for {\sf Label $s$-$t$ Cut}.
Later, Zhang et al. \cite{ZFT16} further refined the algorithms in \cite{TZ12}
to purely combinatorial approximation algorithms (i.e., not using LP-rounding)
for {\sf Label $s$-$t$ Cut} with the same approximation ratios as in \cite{TZ12}.

On the computational hardness side, Zhang et al. \cite{ZCTZ11} showed that
the {\sf Label $s$-$t$ Cut} problem can not be approximated within
%$2^{\log ^{1 - 1/\log\log ^c |\mathcal{I}|} |\mathcal{I}|}$
$2^{(\log |\mathcal{I}|)^{1 - 1/(\log\log |\mathcal{I}|)^c}}$
for any constant $c < 1/2$ unless P $=$ NP, where $|\mathcal{I}|$ is
the input length of the problem. Note that this is a very high hardness factor.
Its order is higher than any polynomial logarithm (i.e., $\log^c n$ for any
constant $c > 0$), but lower than any polynomial (i.e., $n^\epsilon$ for any small
constant $\epsilon > 0$). Essentially the same hardness factor
was independently proved in \cite{CDP+07}.

Fellows et al. \cite{FGK10} considered the parameterized complexity of
the {\sf Label $s$-$t$ Cut} problem. They showed that even in graphs whose
path-width is bounded above by a small constant, the {\sf Label $s$-$t$ Cut}
problem is W[2]-hard when parameterized by the number of used labels.
Recall that W[2] is a class of the W-hierarchy in parameterized complexity.
By the parameterized complexity hypothesis, a problem which is W[$i$]-hard
($i \geq 1$) is not likely fixed-parameter tractable (that is, it is not likely
in FPT).

Jegelka et al. \cite{JB10,JB14} studied a more general cut problem called
{\sf Cooperative $s$-$t$ Cut}, which finds an $s$-$t$ cut such that
an objective function is minimized, where the objective function
can be arbitrary submodular function defined on the edge subsets.
It is not difficult to see that {\sf Cooperative $s$-$t$ Cut} is
a generalization of {\sf Label $s$-$t$ Cut}. Jegelka et al. \cite{JB10,JB14}
gave some approximation algorithms for the {\sf Cooperative $s$-$t$ Cut} problem.

\subsection{Our Results}
In this paper, we study the integrality gaps of two natural linear programming
relaxations for {\sf Label $s$-$t$ Cut}. See (\ref{LP1 - weak LP for label cut})
and (\ref{LP2 - strong LP for label cut}) in the paper.
We prove that both of the two LPs have large integrality gaps.
Let $m$ and $n$ be the edge number and vertex number of the input graph, respectively.
The main theorem of the paper is the following Theorem
\ref{th - integrality gap of LP2 in terms of m}.
\begin{theorem}
\label{th - integrality gap of LP2 in terms of m}
The integrality gap of the LP-relaxation (\ref{LP2 - strong LP for label cut})
is $\Omega(m^{1/3-\epsilon})$, where $\epsilon > 0$ is any small constant.
\end{theorem}

Specifically, we prove that (\ref{LP1 - weak LP for label cut}) has integrality
gap $\Omega(m)$, and (\ref{LP2 - strong LP for label cut}) has has integrality
gap $\Omega(m^{1/3-\epsilon})$ for any small constant $\epsilon > 0$.
Since the graphs we construct for these two results are connected
(implying $m = \Omega(n)$), these two results also imply that
(\ref{LP1 - weak LP for label cut}) has integrality gap $\Omega(n)$,
and (\ref{LP2 - strong LP for label cut}) has has integrality
gap $\Omega(n^{1/3-\epsilon})$ for any small constant $\epsilon > 0$.

Linear program (\ref{LP2 - strong LP for label cut}) is a more stronger version
than (\ref{LP1 - weak LP for label cut}). Our main result is about the integrality
gap of (\ref{LP2 - strong LP for label cut}).
This is proved by a probabilistic
method, that is, we show that with probability larger than zero,
a randomized {\sf Label $s$-$t$ Cut} instance has integrality gap
$\Omega(m^{1/3-\epsilon})$. Therefore, {\em there is} a specific
{\sf Label $s$-$t$ Cut} instance which has integrality gap
$\Omega(m^{1/3-\epsilon})$. Honestly speaking, the proof for this result
is rather complicated. For the sake of readability, we have to write it down
in several separated sections.

Let $\cal I$ be an instance for some minimization problem $\Pi$, and $OPT({\cal I})$
be its optimal value. Let $LP$ be a linear program relaxation for problem $\Pi$,
and $OPT_f(LP({\cal I}))$ be its optimal value on instance $\cal I$.
We use the subscript $f$ to emphasize that $OPT_f(LP({\cal I}))$ is
the {\em fractional} optimal value of $LP$ on instance $\cal I$.
Recall that for a minimized linear program relaxation such as $LP$,
its integrality gap is defined to be the supremum of the ratio between
$OPT({\cal I})$ and $OPT_f(LP({\cal I}))$ over all instances $\cal I$, i.e.,
the integrality gap is
\[
\sup_{\cal I} \left\{\frac{OPT({\cal I})}{OPT_f(LP({\cal I}))}\right\}.
\]

Linear programming is a powerful and successful technique to design approximation
algorithms for NP-hard problems. Some reasons are that,
linear program is in polynomial time solvable and
$OPT_f(LP({\cal I}))$ provides a natural lower bound on $OPT({\cal I})$,
facilitating the design and analysis of approximation algorithms.
On the other hand, from the definition of integrality gap we should learn that,
any approximation algorithm that only use $OPT_f(LP({\cal I}))$
as the lower bound on $OPT({\cal I})$, cannot admit a ratio better than
the integrality gap.

The meaning of our results is then clear: Our results provide
lower bound on the approximation ratios of any approximation algorithms
that are only based on (\ref{LP1 - weak LP for label cut}) or
(\ref{LP2 - strong LP for label cut}) (e.g., the LP-rounding approximation
algorithms and the primal-dual approximation algorithms).
For the {\sf Label $s$-$t$ Cut} problem, if an approximation algorithm only uses
$OPT_f(\ref{LP1 - weak LP for label cut})$ as the lower bound on $OPT$,
then it cannot has an approximation ratio better than $\Omega(m)$.
Similarly, if an approximation algorithm for {\sf Label $s$-$t$ Cut} only uses
$OPT_f(\ref{LP2 - strong LP for label cut})$ as the lower bound on $OPT$,
then it cannot has an approximation ratio better than $\Omega(m^{1/3-\epsilon})$.
These theoretical negative results suggest that to obtain better approximation
ratios for the {\sf Label $s$-$t$ Cut} problem, one should seek new algorithms
other than pure linear programming algorithms.

This paper is the full version of the integrality gap results
in the preliminary conference paper \cite{TZ12}.
A preliminary version of the integrality gap results and their sketch
proofs were given in \cite{TZ12} (in three and half pages).

Besides the integrality gap results, \cite{TZ12} also gave
an $O(\frac{m^{1/2}}{OPT^{1/2}})$ approximation and
an $O(\frac{n^{2/3}}{OPT^{1/3}})$-approximation for the {\sf Label $s$-$t$ Cut}
problem, using a two-stage strategy of LP-rounding and min cut.
After the conference paper \cite{TZ12} was published, we are able to simplify
the approximation algorithms in \cite{TZ12}, getting two purely combinatorial
(i.e., not using LP-rounding) approximation algorithms for {\sf Label $s$-$t$ Cut}
with the same approximation ratios. These algorithmic results are published
in a separate paper (\cite{ZFT16}).

\subsection{More Related Work}
%{\bf Related work for {\sf Global Label Cut}.}
A closely related problem to {\sf Label $s$-$t$ Cut} is the {\sf Min Global
Label Cut} problem ({\sf Global Label Cut} for short). Give an edge-labeled graph,
{\sf Global Label Cut} asks to find a minimum size label set such that
the removal of edges with these labels disconnects the input graph
(into at least two parts). It is easy to see that the {\sf Global Label Cut}
problem is a generalization of the classic {\sf Global Min Cut} problem \cite{KS96}
and the connectivity concept in graph theory.

Zhang et al. \cite{ZCTZ11} first proposed
the {\sf Global Label Cut} problem. They show that this problem can be approximated
within the same factor of {\sf Label $s$-$t$ Cut} by reducing {\sf Global Label Cut}
to {\sf Label $s$-$t$ Cut}. In \cite{ZF16}, Zhang et al. showed that
{\sf Global Label Cut} is polynomial-time solvable for some special types of graphs.
However, the exact complexity (P or NP-hard) of {\sf Global Label Cut} is still
unknown until now.

Very recently, Ghaffari et al. \cite{GKP17} proposed a randomized PTAS
for {\sf Global Label Cut}, where the authors called the problem the {\sf Min Hedge Cut}
problem. Their strategy is the simple but powerful edge contraction technique
developed in \cite{KS96}. Given any small constant $\epsilon > 0$,
in $O(n^{O(\log 1/\epsilon)})$ time, the algorithm in \cite{GKP17}
finds a $(1+\epsilon)$-approximation for {\sf Global Label Cut} with high probability.

Some experimental studies on {\sf Global Label Cut} have also been carried out.
Silva et al. \cite{SSO+16} designed exact algorithms for {\sf Global Label Cut}
using the branch-and-cut and branch-and-bound approaches based on integer programming
formulations for the problem. Bordini et al. \cite{BP17} designed exact algorithms
for {\sf Global Label Cut} using the variable neighborhood search technique.
Both of the authors \cite{SSO+16,BP17} evaluated their algorithms on many
concrete instances of the problem.

\bigskip

{\bf Notations.}
For the ease of statements, some commonly used notations are explained here.
For an input graph $G$, we use $n$ to denote its vertex number, and $m$ its
edge number. Given an instance $\cal I$ of an optimization problem such
as {\sf Label $s$-$t$ Cut},
we use $OPT({\cal I})$ to denote the optimal value of instance $\cal I$.
When $\cal I$ is known from the context, we simply use $OPT$ to denote
$OPT({\cal I})$.

In the {\sf Label $s$-$t$ Cut} problem, given an edge set $E'$, we use $L(E')$
to denote the set of labels appearing in $E'$. Note that $L$ also denotes
the label set in the {\sf Label $s$-$t$ Cut} problem.
We do not introduce more symbols to distinguish these two cases,
just keeping them simple and easily understandable.
Given an edge $e$, we use $\ell(e)$ to denote the label of $e$
(in this case $\ell$ is a mapping from $E(G)$ to $L$). Note that we also
write $\ell \in L$ and in this case $\ell$ denotes some label in $L$.
For simplicity, we do not introduce more symbols to distinguish
these two cases.

For clarity, we use the symbol ``:='' to define notations, and use the symbol
``='' to express equality.

{\bf Organization of the remainder of the paper.}
The remainder of the paper is organized as follows.
In Section \ref{sec - LP1}, we give the first linear program relaxation
(\ref{LP1 - weak LP for label cut}) for {\sf Label $s$-$t$ Cut} and
prove that its integrality gap is $\Omega(m)$.
In Section \ref{sec - LP2}, we give the second linear program relaxation
(\ref{LP2 - strong LP for label cut}) for {\sf Label $s$-$t$ Cut}.
Then the following three sections are used to analyze the integrality gap
of (\ref{LP2 - strong LP for label cut}).
In Section \ref{sec - construction of the instance}, we show the construction
of the {\sf Label $s$-$t$ Cut} instance used to prove the integrality
gap. In Section \ref{sec - high-level idea and the main theorem},
we depict the high-level idea of the proof and give the main theorem of
this paper. In Section \ref{sec - analysis of the integrality gap}, we show
the proof details of the integrality gap of (\ref{LP2 - strong LP for label cut}).
Finally, we conclude the paper in Section \ref{sec - conclusion and discussion}.

\section{A Linear Program and Its Integrality Gap}
\label{sec - LP1}
The following linear program (\ref{LP1 - weak LP for label cut})
is an LP-relaxation for {\sf Label $s$-$t$ Cut}.
In constraint (\ref{lp1 - cut constraint}), $\mathcal{P}_{st}$ denotes
the set of all simple $s$-$t$ paths in $G$,
where an $s$-$t$ path $P$ is viewed as a set of edges in that path.

\begin{alignat}{2}
\label{LP1 - weak LP for label cut}
\min \quad & \sum_{\ell\in L}x_{\ell} & & \tag{LP1} \\
\mbox{s.t.} \quad
\label{lp1 - cut constraint}
& \sum_{e \in P} x_{\ell(e)} \geq 1,
    & &\forall P \in \mathcal{P}_{st}\\
& x_{\ell} \geq 0,
    &\quad& \forall \ell \in L \nonumber
\end{alignat}

To see that (\ref{LP1 - weak LP for label cut}) is an LP-relaxation
for {\sf Label $s$-$t$ Cut}, consider its 0-1 integer version.
Given an instance of {\sf Label $s$-$t$ Cut}, we define
a variable $x_{\ell}\in\{0,1\}$ for each label $\ell \in L$.
The value of $x_{\ell}$ being 1 means that label $\ell$ is chosen and
its value being 0 means not. Constraint (\ref{lp1 - cut constraint}) is
to make sure that for every $s$-$t$ path $P$ in $G$, at least one label
from the edges of $P$ is chosen. Then the set of labels with $x_{\ell} = 1$
forms a solution to the problem.

It is easy to prove that (\ref{LP1 - weak LP for label cut}) has integrality
gap $\Omega(m)$.

\begin{theorem}
Linear program (\ref{LP1 - weak LP for label cut}) has integrality gap $\Omega(m)$.
\end{theorem}
\begin{proof}
Consider the following instance. The graph $G$ (can be
either directed or undirected) is just an $s$-$t$ path of length $n-1$.
The label set $L$ contains only one label $\ell$. Each edge on the path is labeled
with this unique label. Then it is easy to verify that $x_{\ell} = \frac{1}{m}$
is a feasible solution to (\ref{LP1 - weak LP for label cut}) with objective value
$\frac{1}{m}$, while the optimal solution to the instance has value 1.
\end{proof}

\section{A More Strengthened Linear Program}
\label{sec - LP2}
A natural idea to strengthen (\ref{LP1 - weak LP for label cut})
is to sum $x_\ell$ in constraint (\ref{lp1 - cut constraint})
over labels in $L(P)$, instead of over edges in $P$. Thus we get
the following LP-relaxation (\ref{LP2 - strong LP for label cut})
for {\sf Label $s$-$t$ Cut}.

\begin{alignat}{2}
\label{LP2 - strong LP for label cut}
\min \quad & \sum_{\ell\in L}x_{\ell} & & \tag{LP2} \\
\mbox{s.t.} \quad
\label{lp2 - cut constraint}
& \sum_{\ell \in L(P)} x_{\ell} \geq 1,
    & &\forall P \in \mathcal{P}_{st}\\
& x_{\ell} \geq 0,
    &\quad& \forall \ell \in L \nonumber
\end{alignat}

%{\bf Remarks.}
%As in the case of (\ref{LP1 - label cut}), in general there may be
%exponentially many inequalities in constraint (\ref{lp2 - cut constraint}).
%For (\ref{LP2 - strong LP for label cut}), we currently can not find
%a polynomial time separation oracle that deals with constraint (2) in the
%program. If we insist to use an algorithm that can find an $s$-$t$ path
%such that the number of labels in the path is minimized,
%then this idea does not work (assuming P $\neq$ NP),
%since the problem of finding such a path is the {\sf Label Path}
%problem (\cite{HMS07}), which is already known to be NP-hard.
%So, at the current time we do not yet know
%how to solve (\ref{LP2 - strong LP for label cut}) in polynomial time.

Linear program (\ref{LP2 - strong LP for label cut}) is stronger than
(\ref{LP1 - weak LP for label cut}). Any feasible solution
to (\ref{LP2 - strong LP for label cut}) is still feasible
to (\ref{LP1 - weak LP for label cut}), but the opposite direction may not hold.
So, the integrality gap of (\ref{LP2 - strong LP for label cut})
should be hopefully smaller than that of (\ref{LP1 - weak LP for label cut}).
However, we prove that (\ref{LP2 - strong LP for label cut}) has still large
integrality gap $\Omega(m^{1/3-\epsilon})$, where $\epsilon > 0$ is any small constant.
The analysis is rather complicated and we have to defer it to several separated
sections later.

Our analysis of the integrality gap of (\ref{LP2 - strong LP for label cut})
%,especially the construction of the random mappings in the chains
%of the gadget in Figure \ref{fig - gadget shutter},
is inspired by the idea from Charikar et al. \cite{CHK11},
who proved the integrality gap $\Omega(n^{1/3-\epsilon})$
of their linear programming relaxation for a variant of the {\sf Min Label Cover}
problem \cite{AL97}. We follow the analysis framework of \cite{CHK11}.
However, our instance construction is different to \cite{CHK11} and more complicated.
Consequently, in the proof we need more complicated analysis.

In the following we first show how to construct the {\sf Label $s$-$t$ Cut} instance
used in the analysis of integrality gap in
Section \ref{sec - construction of the instance}.
After knowing how the instance is constructed, it is easy to depict the high-level
idea of the analysis, which is done in
Section \ref{sec - high-level idea and the main theorem}. The details of the
analysis is given Section \ref{sec - analysis of the integrality gap}.

In the analysis we shall use many symbols and notations.
To facilitate the reading, we list them in Table \ref{tab - notations}.

\begin{table}
\begin{center}
\begin{tabular}{c|l}
\hline
Notation & Meaning \\
\hline
\hline
$\Phi$ & Ground set of elements \\
\hline
$\mu,\nu$ & Elements in $\Phi$ \\
\hline
$k$ & Number of elements in $\Phi$ \\
\hline
$d$ & Number of diamonds in a chain \\
\hline
$\sigma$ & Random mapping from $[d]$ to $[d]$ \\
\hline
$H_{\mu\nu}$ & Shutter of $\mu$ and $\nu$ \\
\hline
$h$ & Number of chains in a shutter \\
\hline
$C^i_{\mu\nu}$ & The $i$-th chain in shutter $H_{\mu\nu}$ \\
\hline
$\sigma^i_{\mu\nu}$ & The random mapping on chain $C^i_{\mu\nu}$ \\
\hline
$L'$ & Label subset \\
\hline
$c$ & Size of $L'$ \\
\hline
$J_{\mu}$ & Set of $j$'s such that $(\mu, j) \in L'$ \\
\hline
$a$ & Average of $|J_{\mu}|$'s \\
\hline
$\Phi'$ & Light ground set of $L'$ \\
        & (i.e., set of $\mu$'s such that $|J_{\mu}| \leq 4a$) \\
\hline
${\cal H}_{\Phi'}$ & The set all shutters for every two ordered \\
                   & pair in $\Phi'$ \\
\hline
$F$ & Configuration of $L'$ \\
\hline
$z$ & The upper bound of the probability that \\
    & there exists a good configuration for $\cal I$ \\
\hline
$\epsilon$ & Any given positive small constant \\
\hline
$\beta,\delta$ & Two constants depending on $\epsilon$ \\
\hline
\end{tabular}
\vspace{2pt}
\caption{Main notations used in the analysis of (\ref{LP2 - strong LP for label cut})}
\label{tab - notations}
\end{center}
\end{table}

\section{Construction of the Instance}
\label{sec - construction of the instance}
Let $k$, $d$ and $h$ be three integer parameters that will be determined later.
Define
\begin{equation}
\label{eqn - def of U}
\Phi := \{1, 2, \ldots, k\}
\end{equation}
as a {\em ground set} of $k$ {\em elements}.

{\em Remarks.}
The values of $d$ and $h$ are given in (\ref{eqn - def of d}) and
(\ref{eqn - def of h}), which are both functions of $k$. As for $k$, we only need
it to be a sufficiently large integer. The specific requirement on $k$
(how large $k$ should be) is given in (\ref{eqn - how large k should be}).

\subsection{The Chain Gadget}
\label{sec - the chain gadget}
First we introduce the chain gadget as shown in Figure \ref{fig - gadget chain}.
This gadget will be repeatedly used in the construction of the {\sf Label $s$-$t$ Cut}
instance.

\begin{figure}
\begin{center}
\includegraphics*[width=0.85\textwidth]{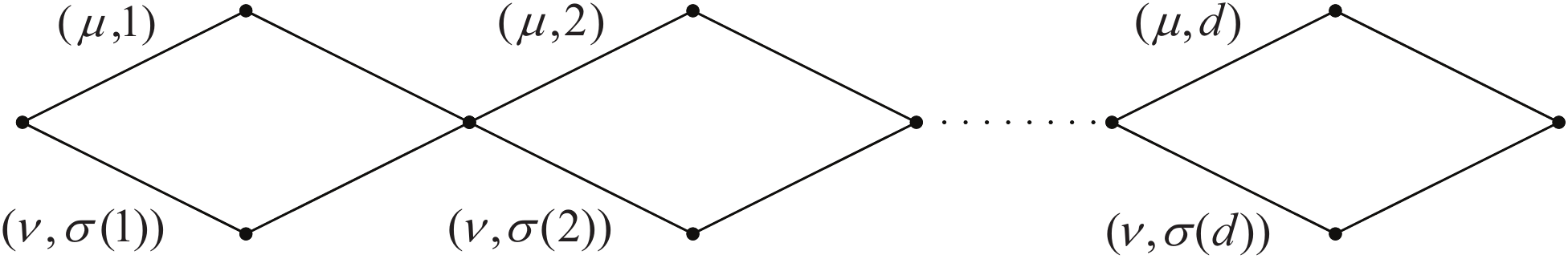}
\end{center}
\caption{The chain gadget.
\label{fig - gadget chain}}
\end{figure}

A chain is constituted of $d$ diamonds. By diamond we mean
a simple cycle of four edges, with two of them being {\em top edges} and
the other two being {\em bottom edges}, as shown in Figure \ref{fig - gadget chain}.
In a chain, two elements $\mu$ and $\nu$ from the ground set $U$ will be used
to constitute labels on edges. Every label is of the form $(\mu, j)$, where $\mu$
is an element from $\Phi$ and $j$ is an integer from $[d]$ ($[d]$ denotes the set
$\{1, 2, \ldots, d\}$). Besides, we will use a random mapping
\[
    \sigma \colon [d] \rightarrow [d],
\]
which is a permutation drawn uniformly at random.
We would like to say that the random permutation $\sigma$ plays an important role
in the analysis of the integrality gap of (\ref{LP2 - strong LP for label cut}).
We will interchangeably use permutation and mapping for $\sigma$.

In the $j$-th diamond for each $j \in [d]$, the two top edges are labeled
with label $(\mu, j)$, and the two bottom edges are labeled with $(\nu, \sigma(j))$.
For clarity, in each diamond in Figure \ref{fig - gadget chain}, the labels
on the latter top edge and the latter bottom edge are omitted.

It is then clear why we use diamonds to constitute a chain. We just want to
make the resulting graph being a simple graph. In fact, if multi-edges are allowed,
we could also use 2-edge cycles to constitute a chain.

For a chain, we call the set of all the top edges of all its diamonds
the {\em top half-chain}, and call the set of all the bottom edges of all its diamonds
the {\em bottom half-chain}.
It is important to note that there is a mapping from the second components of labels
on the top half-chain to the second components of labels on the bottom half-chain.
This mapping, is just the random mapping $\sigma$.

\subsection{The Shutter Gadget and the Final Graph}
For each pair of elements $\mu$ and $\nu$ in $\Phi$ such that $\mu < \nu$, we construct
a shutter gadget $H_{\mu\nu}$ as shown in Figure \ref{fig - gadget shutter}.
Shutter $H_{\mu\nu}$ consists of $h$ chains $C_{\mu\nu}^1,$ $C_{\mu\nu}^2,$ $\ldots,$
$C_{\mu\nu}^h$, where each chain is the one constructed in Section \ref{sec - the chain gadget}.
All the left endpoints of the $h$ chains are merged into a single vertex $s_{\mu\nu}$,
while all the right endpoints of the $h$ chain are merged into a single vertex $t_{\mu\nu}$.

\begin{figure}
\begin{center}
\includegraphics*[width=1\textwidth]{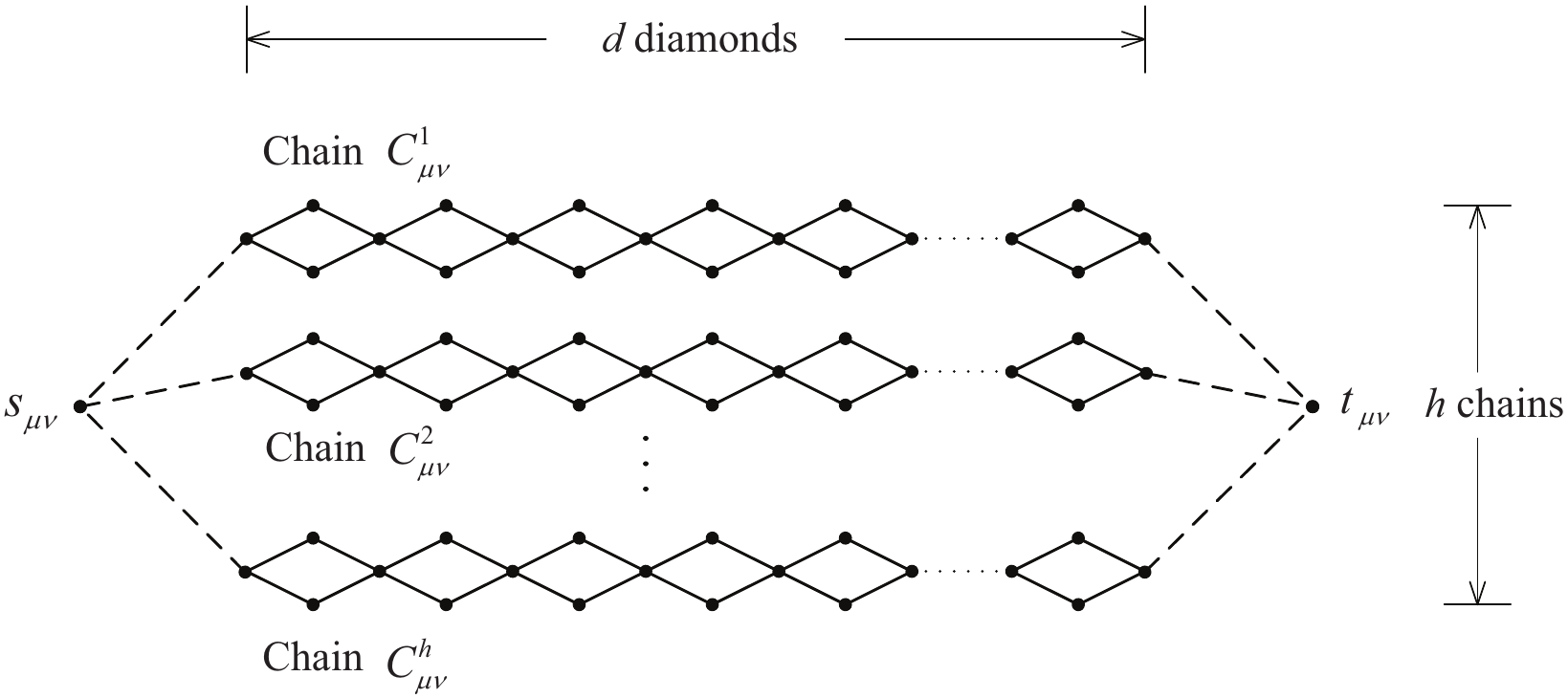}
\end{center}
\caption{The gadget $H_{\mu\nu}$. The dashed lines means that the two endpoints
of each of them are actually merged into a single vertex.
\label{fig - gadget shutter}}
\end{figure}

Note that in the shutter gadget we have $h$ independent random permutations,
denoted by $\sigma_{\mu\nu}^1$, $\sigma_{\mu\nu}^2$, $\ldots$, $\sigma_{\mu\nu}^h$.
Therefore, the only difference between two chains of a shutter is the difference
of their labels.
More specifically, the only difference between two chains of a shutter is
the difference of the labels of their bottom half-chains.
In a shutter, all top half-chains have the same set of labels.

Given the $k \choose 2$ shutter gadgets $H_{11}$, $H_{12}$, $\ldots$, $H_{k-1,k}$
constructed as above, we merge all the left endpoints of these shutters
into a single vertex, which is the source vertex $s$. Similarly, we merge
all the right endpoints of these shutters into a single vertex,
which is the sink vertex $t$. This is our final graph $G$, as shown
in Figure \ref{fig - graph G}. It is easy to see that graph $G$ can be made
directed by orienting all its edges from $s$ to $t$.
Note that all the random permutations appeared in $G$ are independent.

\begin{figure}
\begin{center}
\includegraphics*[width=0.85\textwidth]{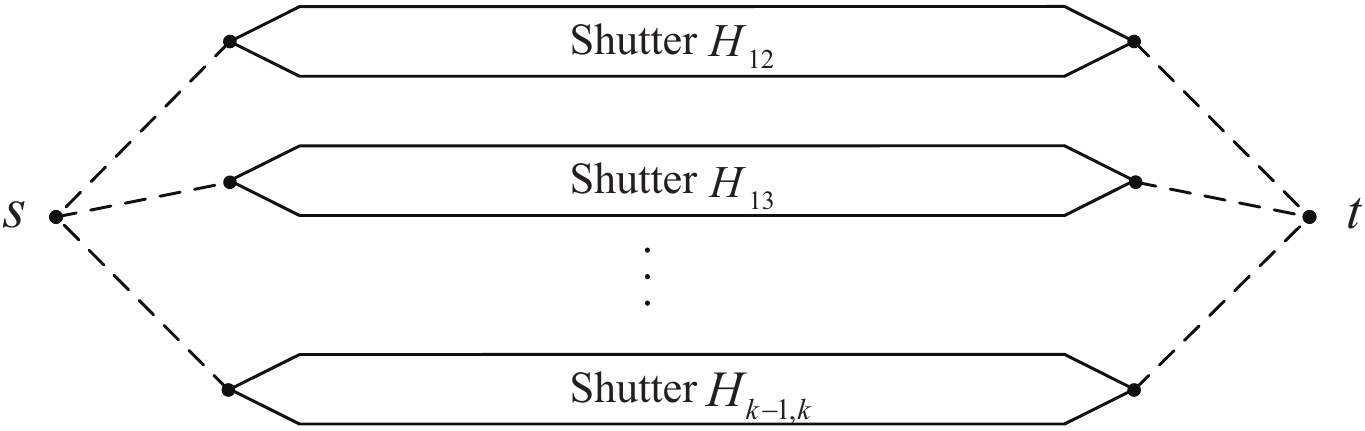}
\end{center}
\caption{Graph $G$. The dashed lines means that the two endpoints
of each of them are actually merged into a single vertex.
\label{fig - graph G}}
\end{figure}

At last, let
\[
L := \{(\mu,j) \colon \mu \in \Phi, j \in [d] \}.
\]
Thus we get the random {\sf Label $s$-$t$ Cut} instance ${\cal I} := (G, s, t, L)$.
By the construction, we know that
\begin{eqnarray}
\label{eqn - n = Theta(k^2 dh)}
&{}& n := |V(G)|=\Theta(k^{2} d h), \\
&{}& m := |E(G)|=\Theta(k^{2} d h), \nonumber \\
\label{eqn - q = kd}
&{}& q := |L| = k d.
\end{eqnarray}

\section{High-level Idea to
Analyze the Integrality Gap of (\ref{LP2 - strong LP for label cut})
and the Main Theorem}
\label{sec - high-level idea and the main theorem}

\subsection{The High-level Idea}
After we have known how to construct the {\sf Label $s$-$t$ Cut} instance $\cal I$,
it is now appropriate to state the high-level idea to prove that
linear program (\ref{LP2 - strong LP for label cut}) has large integrality gap.
The instance we have just constructed is a {\em random instance}.
Note that in the instance we use random permutations to generate labels
for all the chains, and the random permutations are independent and uniform at random.
This fact will play an important role in our analysis.

Let us fix a positive number $c$ which is the size of a presumed solution
to the random instance $\cal I$. The overall strategy is to prove that
there exists a fixed instance (i.e., sample) $\hat{\cal I}$ of random instance
$\cal I$, for which {\em any} presumed solution of the given size $c$
is {\em not feasible}. This means that instance $\hat{\cal I}$ has relatively
large integral optimum (i.e., $OPT(\hat{\cal I}) \geq c$).
This is the technical result of this paper, which is formally stated
in the following Lemma \ref{lm - Integral OPT of LP2 is large}.
Meanwhile, it is not difficult to prove that the fractional
optimum of (\ref{LP2 - strong LP for label cut}) on instance $\hat{\cal I}$
(i.e., $OPT_f(\ref{LP2 - strong LP for label cut}(\hat{\cal I}))$) is relatively small.
Consequently, a large integrality gap of (\ref{LP2 - strong LP for label cut})
is concluded by carefully choosing the parameters in the instance construction.

\begin{lemma}[The Technical Lemma]
\label{lm - Integral OPT of LP2 is large}
For any small constant $\epsilon > 0$, there exists a constant $k_0$
which depends only on $\epsilon$, such that for any integer $k \geq k_0$,
there exists a {\sf Label $s$-$t$ Cut} instance $\hat{\cal I}$
whose minimum label cut is of size $\Omega(kn^{1/3-\epsilon})$.
\end{lemma}

We shall prove Lemma \ref{lm - Integral OPT of LP2 is large}
in Section \ref{sec - analysis of the integrality gap}.
Here we show the idea of the proof. Let $c > 0$ be a number we will fix later,
and $L' \subseteq L$ be any label subset of size $c$. We show that
there exist an element subset $\Phi' \subseteq \Phi$ determined by $L'$,
and a set ${\cal H}_{\Phi'}$ of shutters determined in turn by $\Phi'$, such that
(i) ${\cal H}_{\Phi'}$ consists of large number of shutters, and
(ii) for each shutter in ${\cal H}_{\Phi'}$, $L'$ only contains bounded number
of labels in the shutter. Since the random mapping of labels on every
chain in each shutter is drawn independently, the probability
that $s$ and $t$ are separated in ${\cal H}_{\Phi'}$ by $L'$ is very small.
Consequently, for a particularly specified but still large number $c$,
there exists a fixed instance (i.e., sample) $\hat{\cal I}$ of the random
instance $\cal I$, such that any $L'$ of size $c$ cannot separate $s$ and $t$
in the corresponding ${\cal H}_{\Phi'}$ of $\hat{\cal I}$.
So, we get a large lower bound on the optimal value of instance $\hat{\cal I}$,
that is, $OPT(\hat{\cal I}) \geq c$.

\subsection{The Main Theorem}
\begin{lemma}
\label{lm - Fractional OPT of LP2 is small}
For any fixed instance (i.e., sample) ${\cal I}'$ of the random {\sf Label $s$-$t$ Cut}
instance $\cal I$ constructed in Section \ref{sec - construction of the instance},
we have
\[
\text{OPT}_f(\ref{LP2 - strong LP for label cut}({\cal I}')) \leq k,
\]
where $\text{OPT}_f(LP2({\cal I}'))$ is the fractional optimum of
(\ref{LP2 - strong LP for label cut}) on instance ${\cal I}'$.
\end{lemma}
\begin{proof}
For each label $(u,j) \in L$, we assign $x_{(u,j)} = 1/d$.
The only constraint (\ref{lp2 - cut constraint}), that is,
$\sum_{\ell \in L(P)} x_{\ell} \geq 1$ for any $s$-$t$ path $P$, is satisfied
since any simple $s$-$t$ path in graph $G$ contains exactly $d$ distinct labels.
So, $x$ is a feasible solution to (\ref{LP2 - strong LP for label cut}),
whose objective value is $\sum_{(u,j)\in L} x_{(u,j)} = |L|/d = k$.
This implies $\text{OPT}_f(\ref{LP2 - strong LP for label cut}({\cal I}')) \leq k$.
\end{proof}

With the help of Lemma \ref{lm - Integral OPT of LP2 is large} and
Lemma \ref{lm - Fractional OPT of LP2 is small}, it is easy to prove
the main theorem.

\begin{theorem}
\label{th - integrality gap of LP2}
The integrality gap of the LP-relaxation (\ref{LP2 - strong LP for label cut})
is $\Omega(n^{1/3-\epsilon})$, where $\epsilon > 0$ is any small constant.
\end{theorem}
\begin{proof}
Let us consider (\ref{LP2 - strong LP for label cut})
on instance $\hat{\cal I}$ given in Lemma \ref{lm - Integral OPT of LP2 is large}.

By Lemma \ref{lm - Fractional OPT of LP2 is small}, we have
$\text{OPT}_f(\ref{LP2 - strong LP for label cut}(\hat{\cal I})) \leq k$,
where $\text{OPT}_f(\ref{LP2 - strong LP for label cut}(\hat{\cal I}))$
is the fractional optimum of (\ref{LP2 - strong LP for label cut})
on instance $\hat{\cal I}$. By Lemma \ref{lm - Integral OPT of LP2 is large},
the integrality gap of (\ref{LP2 - strong LP for label cut})
on instance $\hat{\cal I}$
is
\[
\frac{\text{OPT}(\hat{\cal I})}{\text{OPT}_f(LP2)} \geq \frac{\Omega(k n^{1/3-\epsilon})}{k}
= \Omega(n^{1/3-\epsilon}).
\]
%\qed
\end{proof}

%
% ======================= The Main Theorem =======================
%

\noindent
{\bf Theorem \ref{th - integrality gap of LP2 in terms of m}.}
(restated)
\begin{em}
The integrality gap of the LP-relaxation (\ref{LP2 - strong LP for label cut})
is $\Omega(m^{1/3-\epsilon})$, where $\epsilon > 0$ is any small constant.
\end{em}
\begin{proof}
By Theorem \ref{th - integrality gap of LP2} and the fact that $m = \Theta(n)$
for the constructed instance $\cal I$.
\end{proof}

\section{Analysis of the Integrality Gap -- the Details}
\label{sec - analysis of the integrality gap}
Let $c > 0$ be an integer which denotes the size of a label subset of $L$.
We shall show that for a particularly chosen value of $c$
(see (\ref{eqn - def of c})), there exists a fixed
instance $\hat{\cal I}$ of the random {\sf Label $s$-$t$ Cut} instance $\cal I$,
for which no label cut of size $c$ exists. Thus we infer a lower bound
(see (\ref{eqn - lower bound on c})) on the size of the minimum
label cut of $\hat{\cal I}$.

\subsection{Structure of the Solution}
Let $L' \subseteq L$ be any label subset of size $c$.
$L'$ will be used as a solution to {\sf Label $s$-$t$ Cut}, but it may
not be feasible. Recall that $\Phi$ is the ground set (see (\ref{eqn - def of U})).
Define
\[
    J_{\mu} := \{j \in [d] \colon (\mu,j) \in L' \}, \qquad \forall \mu \in \Phi.
\]
Then,
\begin{equation}
\label{eqn - c = sum of c_mu}
    c = |L'| = \sum_{\mu \in \Phi}|J_{\mu}|.
\end{equation}
Moreover, define
\begin{equation}
\label{eqn - definition of a}
a := \frac{c}{k}.
\end{equation}
Then $a$ is the average of $|J_{\mu}|$'s.

For an element $\mu \in \Phi$, if $|J_{\mu}| \leq 4a$, then $\mu$ is called
a {\em light} element (which means that it appears not heavily in $L'$).
Otherwise $\mu$ is called a {\em heavy} element.

We further define
\begin{equation}
\label{eqn - Phi'}
    \Phi' := \{\mu \in \Phi \colon |J_{\mu}| \leq 4a \}.
\end{equation}
$\Phi'$ is called the {\em light ground set} with respect to $L'$.
For each element $\mu \in \Phi'$, there are $|J_{\mu}| \leq 4a$ labels in $L'$
related to $\mu$. Note that $\Phi'$ also contains the elements $\mu$
that $|J_{\mu}| = 0$ (if there are).

By (\ref{eqn - c = sum of c_mu}), the number of
elements $\mu$ such that $|J_{\mu}| > 4a$ is at most $k/4$. This implies
that the number of elements $\mu$ such that $|J_{\mu}| \leq 4a$ is at least $3k/4$.
That is, we have
\begin{equation}
\label{eqn - |Phi'| >= 3k/4}
    |\Phi'| \geq \frac34 k.
\end{equation}
Thus $\Phi'$ contains most elements in $\Phi$.
$\Phi'$ is called a light ground set in the sense that each element in $\Phi'$
appears not heavily in $L'$.

\begin{definition}
{\bf Solution configuration.}

Given a label subset, its light ground set is defined accordingly
(as in (\ref{eqn - Phi'})). The set of all the labels in this label subset
that are related to some element in its light ground set is called
a {\em solution configuration} (configuration for short).
\end{definition}

Let $F$ be the configuration determined by $L'$. Then we have
\[
    F = \{(\mu, j) \in L' \colon \mu \in \Phi'\}.
\]
By definition, configuration $F$ possesses the following property:
For every element $\mu$ that appears in $F$, the number of labels in $F$
which are related to $\mu$ is at most $4a$.
Note that different solutions (they may not be feasible in general)
may lead to the same configuration.

\bigskip

{\bf Some further explanations for $L'$, $\Phi'$, and $F$.}
If we define
\[
    \Phi_{L'} := \{\mu \in \Phi \colon \exists j, (\mu, j) \in L' \},
\]
then in general we may not have $\Phi_{L'} \subseteq \Phi'$,
since $\Phi_{L'}$ may contain a heavy element while every element in $\Phi'$
is light. Similarly, in general we also may not have $\Phi' \subseteq \Phi_{L'}$,
since $\Phi'$ may contain an element $\mu$ with $|J_{\mu}| = 0$,
while $|J_{\mu}| = 0$ means that $\mu$ is not in $\Phi_{L'}$ at all.

If we define
\[
    \Phi_F := \{\mu \in \Phi \colon \exists j, (\mu, j) \in F \},
\]
then naturally we have $\Phi_F \subseteq \Phi_{L'}$. $\Phi_{L'}$ may contain
heavy element(s), while $\Phi_F$ never contain such elements.
However, neither $\Phi_F$ nor $\Phi_{L'}$ contains an element $\mu$
with $|J_{\mu}| = 0$. By definitions, we actually have
\[
    \Phi_F = \Phi' \cap \Phi_{L'}.
\]

See Figure \ref{fig - configuration} for an illustration of configuration $F$.

\begin{figure}
\begin{center}
\includegraphics*[width=0.9\textwidth]{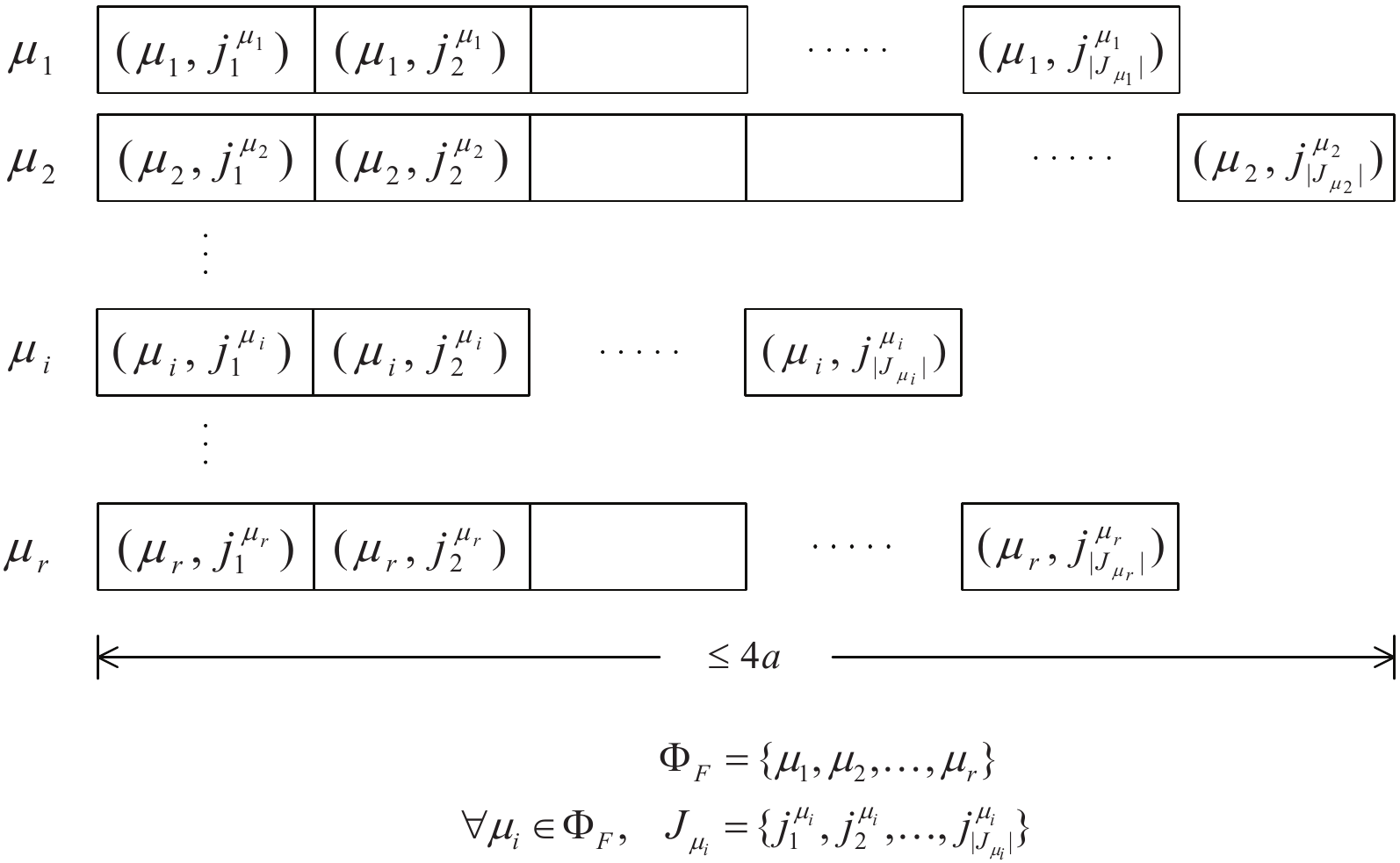}
\end{center}
\caption{An illustration of a solution configuration.
Each block denotes a label.
%We suppose that $\Phi_F = \{\mu_1, \mu_2, \ldots, \mu_g\}$,
%and for each element $\mu_i \in \Phi_F$,
%$J_{\mu_i} = \{j^{\mu_i}_1, j^{\mu_i}_2, \ldots, j^{\mu_i}_{|J_{\mu_i}|}\}$.
\label{fig - configuration}}
\end{figure}

%
% ==============================================================
%

\subsection{Relationship of the Random Mapping and the Configuration}
Let ${\cal H}_{\Phi'}$ be the set of all shutter gadgets $H_{\mu\nu}$
for every pair $\mu, \nu \in \Phi'$ such that $\mu < \nu$, i.e.,
\[
    {\cal H}_{\Phi'} := \{ H_{\mu\nu} \colon \mu, \nu \in \Phi', \mu < \nu\}.
\]
In the following we compute the probability that $s$ and $t$ are separated
in ${\cal H}_{\Phi'}$ by $L'$.

Note that by the definitions of $\Phi'$, ${\cal H}_{\Phi'}$, and configuration $F$,
the labels in $L'$ that appear in ${\cal H}_{\Phi'}$ are just the same as in $F$.
That is,
\[
    F = L' \cap L({\cal H}_{\Phi'}),
\]
where we use $L({\cal H}_{\Phi'})$ to denote the set of all labels appearing
in ${\cal H}_{\Phi'}$.
So, the event that $s$ and $t$ are separated in ${\cal H}_{\Phi'}$ by $L'$
is the same as the event that $s$ and $t$ are separated in ${\cal H}_{\Phi'}$
by $F$. If $F$ separates $s$ and $t$ in ${\cal H}_{\Phi'}$ (that is, $s$ and $t$
are separated in ${\cal H}_{\Phi'}$ by removing all the edges having labels
in $F$), then we call $F$ a {\em good configuration}. Otherwise we call $F$
a {\em bad configuration}. Therefore, the goal of this section can be equivalently
restated as computing the probability that $F$ is a good configuration.

\begin{lemma}
\label{lm - prob of s and t are separated in Cuvi by L'}
Let $H_{\mu\nu}$ be a shutter in ${\cal H}_{\Phi'}$ for some elements $\mu$ and $\nu$,
and $C_{\mu\nu}^i$ be the $i$-th chain in $H_{\mu\nu}$.
The probability that $s$ and $t$ are separated in $C_{\mu\nu}^i$ by $L'$
is at most
\[
    1 - \left( 1 -\frac{8a}{d} \right)^{4a}.
\]
\end{lemma}
\begin{proof}
Since $H_{\mu\nu} \in {\cal H}_{\Phi'}$, we have $\mu \in \Phi'$ and $\nu \in \Phi'$
by definition. This implies that in $L'$ the number of labels of the form
$(\mu, \cdot)$ is at most $4a$. Likewise, the number of labels of the form
$(\nu, \cdot)$ in $L'$ is also at most $4a$. Therefore, we can upper bound
the probability that $s$ and $t$ are separated in $C_{\mu\nu}^i$ by $L'$,
since all the labels appeared in $C_{\mu\nu}^i$ are of the forms $(\mu, \cdot)$
or $(\nu, \cdot)$, and $L'$ uses bounded number of these labels.

Since $C_{\mu\nu}^i$ is a series of consecutive diamonds
(see Figure \ref{fig - gadget chain}), $s$ and $t$ are separated in $C_{\mu\nu}^i$
if and only if there is a diamond in $C_{\mu\nu}^i$ that at least one of its two top edges
and at least one of its two bottom edges are removed.

Suppose that the random mapping $\sigma_{\mu\nu}^i$ maps $j \in J_\mu$ into
$J_\nu$, that is, $\sigma_{\mu\nu}^i(j) \in J_\nu$.
Then, for the $j$-th diamond of the chain $C_{\mu\nu}^i$,
all the labels of its top edges and bottom edges are included in $L'$.
This is because all the labels in $\{(\mu, j) \colon j \in J_\mu\}$ and
$\{(\nu, j) \colon j \in J_\nu\}$ are in $L'$.
Therefore, if the random mapping $\sigma_{\mu\nu}^i$ maps an element in $J_\mu$ into
$J_\nu$, then $s$ and $t$ are separated in $C_{\mu\nu}^i$ by $L'$.
In other words, $s$ and $t$ are not separated in $C_{\mu\nu}^i$ by $L'$
if and only if $\sigma_{\mu\nu}^i$ maps all $j$'s in $J_\mu$ outside $J_\nu$.
That is,
\begin{eqnarray}
&&\Pr[\mbox{$s$ and $t$ are {\em not} separated by $L'$ in $C_{\mu\nu}^i$}]
\nonumber \\
\label{eqn - prob of s and t are separated in Cuvi}
&=& \Pr[\sigma_{\mu\nu}^i(J_{\mu}) \cap J_{\nu} = \emptyset],
\end{eqnarray}
where $\sigma_{\mu\nu}^i(J_{u}) = \{\sigma_{\mu\nu}^i(j) \colon j \in J_u\}$.

Recall that $\sigma_{\mu\nu}^i$ is a mapping from $[d]$ to $[d]$,
$J_\mu \subseteq [d]$, and $J_\nu \subseteq [d]$. Among all the ways
mapping $J_\mu$ to $[d]$ (there are $d \choose |J_{\mu}|$ such ways),
there are $d - |J_{\nu}| \choose |J_{\mu}|$ ways mapping $J_\mu$ to $[d] \setminus J_\nu$
(i.e., outside $J_\nu$). So, we have
\begin{eqnarray*}
\Pr[\sigma_{\mu\nu}^i(J_{\mu}) \cap J_{\nu} = \emptyset]
&=& \frac{ {d - |J_{\nu}| \choose |J_{\mu}|} } { {d \choose |J_{\mu}|} } \\
&\geq& \frac{ {d - 4a \choose 4a} } { {d \choose 4a} } \\
&=& \underbrace{\frac{d - 4a}{d} \cdot \frac{(d - 4a) - 1}{d - 1} \cdots
  \frac{(d - 4a) - (4a - 1)}{d - (4a - 1)}}_{4a \mathrm{~items}} \\
&\geq& \left( 1 -\frac{8a}{d} \right)^{4a}.
\end{eqnarray*}

Therefore, the probability that $s$ and $t$ are separated in $C_{\mu\nu}^i$
by $L'$ is at most $1 - \left( 1 -\frac{8a}{d} \right)^{4a}$.
We remark that $d$ is strictly greater than $8a$ by our later choice of parameters
(see (\ref{eqn - def of d}) and (\ref{eqn - a = k^delta})).
The lemma follows.
\end{proof}

\begin{lemma}
\label{lm - prob of s and t are separated in Hphi' by L'}
\label{lm - prob that configuration is good}
The probability that $F$ is a good configuration (i.e., the probability
that $s$ and $t$ are separated in ${\cal H}_{\Phi'}$ by $L'$)
is at most
\[
    \left[1 - \left( 1 -\frac{8a}{d} \right)^{4a}\right]^{\frac12 h (3k/4)(3k/4 - 1)}.
\]
\end{lemma}
\begin{proof}
Let $H_{\mu\nu}$ be any shutter in ${\cal H}_{\Phi'}$. Since $H_{\mu\nu}$ contains
$h$ chains $C_{\mu\nu}^1$, $\ldots$, $C_{\mu\nu}^h$,
by Lemma \ref{lm - prob of s and t are separated in Cuvi by L'}, we have
\begin{eqnarray}
&& \Pr[\mbox{$s$ and $t$ are separated in $H_{\mu\nu}$ by $L'$}]
\nonumber \\
&=& \Pr[\forall 1 \leq i \leq h \colon
    s \mbox{ and } t \mbox{ are separated in } C_{\mu\nu}^i \mbox{ by } L']
    \nonumber \\
\label{eqn - prob of s and t are separated in Huv}
&\leq& \left[ 1 - \left( 1 -\frac{8a}{d} \right)^{4a} \right]^h.
\end{eqnarray}

So, for ${\cal H}_{\Phi'}$ we have
\begin{eqnarray}
&& \Pr[\text{$s$ and $t$ are separated in ${\cal H}_{\Phi'}$ by $L'$}]
    \nonumber \\
&=& \Pr[\forall \mu,\nu\in \Phi' \text{ s.t. } \mu < \nu,
    \text{$s$ and $t$ are separated in $H_{\mu\nu}$}]
    \nonumber \\
&\stackrel[\text{(\ref{eqn - prob of s and t are separated in Huv})}]{}{\leq}&
    \prod_{\mu,\nu\in \Phi'} \left[ 1 - \left( 1 -\frac{8a}{d} \right)^{4a} \right]^h
    \nonumber \\
\label{eqn - prob. s and t are separated in H_U'}
&\stackrel[\text{(\ref{eqn - |Phi'| >= 3k/4})}]{}{\leq}&
    \left[1 - \left( 1 -\frac{8a}{d} \right)^{4a} \right]^{\frac12 h (3k/4)(3k/4 - 1)}.
\end{eqnarray}
\end{proof}

%
% ==============================================================
%

\subsection{Proof of the Technical Lemma \ref{lm - Integral OPT of LP2 is large}}
We first state the overall strategy to prove
Lemma \ref{lm - Integral OPT of LP2 is large}.
Given a solution of size $c$ to the instance $\cal I$,
we can figure out its corresponding configuration.
Since there are many solutions of size $c$, there are many different
configurations. (Different solutions may lead to the same configuration.)
All these configurations are about instance $\cal I$.
If all these configurations are bad configurations, then any solution
of size $c$ cannot separate $s$ and $t$. So, a feasible solution
to the instance has to have size strictly larger than $c$.

Since $\cal I$ is a random instance, we have to compute
the probability that there exists a good configuration for $\cal I$.
If this probability is less than one, then with non-zero probability,
the random instance $\cal I$ has no good configuration.
So, there exists a fixed instance $\hat{\cal I}$ (i.e., a sample) of
the random instance $\cal I$, for which all its configurations
are bad. Consequently, any feasible solution to instance $\hat{\cal I}$
has size $> c$.

In order to compute the probability that there exists a good configuration
for the random instance $\cal I$, let us first count for the number of
all configurations of $\cal I$.

\begin{lemma}
\label{lm - number of configurations}
Given the solution size $c$, the number of solution configurations of
instance $\cal I$ is at most
\[
    (4a+1)^k d^{(4a+1)k}.
\]
\end{lemma}
\begin{proof}
Since the element set for a configuration is a subset of the ground set $\Phi$,
and there are $k$ elements in $\Phi$, the number of all possible element
sets for configurations is at most
\begin{equation}
\label{eqn - upper bound on the number of element sets}
    2^k.
\end{equation}

Fix an element set of size $r$ (e.g., the set $\Phi_F$ in Figure \ref{fig - configuration}),
and consider an element $\mu$ in this set. There are at most $4a$ labels related
to $\mu$ in a configuration (see a line in Figure \ref{fig - configuration}).
The possibilities that at most $4a$ such labels appear is
${d \choose 1} + {d \choose 2} + \cdots + {d \choose 4a}$, since there are $d$ labels
$\{(u, j) \colon 1 \leq j \leq d\}$ related to $\mu$ in total in $L$.
Since there are $r$ elements in the element set, the total possibilities
of configurations for this fixed element set is
\begin{equation}
\label{eqn - [...]^r}
    \left[{d \choose 1} + {d \choose 2} + \cdots + {d \choose 4a}\right]^r.
\end{equation}

By \cite[Equation (5.18)]{GKP94}, we have
${d \choose 1} + {d \choose 2} + \cdots + {d \choose 4a}
\leq
\frac{4a+1}{2} {d \choose {4a+1}}$.
So, the value of (\ref{eqn - [...]^r}) is at most
\begin{equation}
\label{eqn - upper bound on the number of cfgs for fixed element set}
\left[\frac{4a+1}{2} {d \choose {4a+1}}\right]^r
\leq
\left[\frac{4a+1}{2} {d \choose {4a+1}}\right]^k
\leq
\left[ \frac12 (4a+1) d^{4a+1} \right]^k.
\end{equation}

By (\ref{eqn - upper bound on the number of element sets})
and (\ref{eqn - upper bound on the number of cfgs for fixed element set}),
the total number of configurations of instance $\cal I$ is at most
$(4a+1)^k d^{(4a+1)k}$, proving the lemma.
\end{proof}

\bigskip

{\bf The probability that there exists a good configuration for instance $\cal I$.}
Define
\begin{equation}
\label{eqn - def of z}
    z := (4a+1)^k d^{(4a+1)k}
        \left[1 - \left( 1 -\frac{8a}{d} \right)^{4a}\right]^{\frac12 h (3k/4)(3k/4 - 1)}.
\end{equation}
Then by Lemma \ref{lm - prob that configuration is good} and
Lemma \ref{lm - number of configurations},
$z$ is an upper bound of the probability that there exists a good configuration
for instance $\cal I$. We can choose the values of $d$, $h$ and $c$
(recall that $a = c/k$) so that $z < 1$ (see Lemma \ref{lm - z < 1}), and hence
the technical Lemma \ref{lm - Integral OPT of LP2 is large} can be proved.

\bigskip

{\bf Settling the values of $d$, $h$, and $c$.}
Now we settle the values of the parameters. Let $\epsilon$ be any small constant
in $(0, 1/3)$, and $\delta>0$ and $\beta>0$ are two constants whose values only
depend on $\epsilon$. We define
\begin{eqnarray}
\label{eqn - def of d}
d &:=& 32{k^{2\delta}}, \\
\label{eqn - def of h}
h &:=& k^{\beta}, \\
\label{eqn - def of c}
c &:=& k^{1+\delta}.
\end{eqnarray}
The values of $\delta$ and $\beta$ will be given in Section \ref{sec - z < 1}.

All we have done until now will be put together in the following
Lemma \ref{lm - z < 1} and Lemma \ref{lm - c = Omega(kn^(1/3-epsilon))}.
Their proofs are deferred to Section \ref{sec - z < 1}.

\begin{lemma}
\label{lm - z < 1}
For any positive constants $\beta$ and $\delta$, as long as
\begin{equation}
\label{eqn - condition that beta > delta-1}
    \beta > \delta - 1,
\end{equation}
we will have
\[
    z < 1
\]
for large enough $k$, that is, for any $k \geq k_0$,
where $k_0$ is a constant depending only on $\beta$ and $\delta$.
\end{lemma}

\begin{lemma}
\label{lm - c = Omega(kn^(1/3-epsilon))}
Let $\epsilon$ be any small constant in $(0, 1/3)$. There exist
appropriate positive values of $\delta$ and $\beta$, which only depend on $\epsilon$,
such that
\[
    c = \Omega(k n^{1/3 - \epsilon}).
\]
\end{lemma}

With the help of Lemma \ref{lm - z < 1} and Lemma \ref{lm - c = Omega(kn^(1/3-epsilon))},
it is easy to prove Lemma \ref{lm - Integral OPT of LP2 is large}.

\bigskip

\noindent
{\bf Lemma \ref{lm - Integral OPT of LP2 is large}.}
(restated)
\begin{em}
For any small constant $\epsilon > 0$, there exists a constant $k_0$
which depends only on $\epsilon$, such that for any integer $k \geq k_0$,
there exists a {\sf Label $s$-$t$ Cut} instance $\hat{\cal I}$
whose minimum label cut is of size $\Omega(kn^{1/3-\epsilon})$.
\end{em}
\begin{proof}[Proof of Lemma \ref{lm - Integral OPT of LP2 is large}]
By Lemma \ref{lm - z < 1}, the probability that there exists a good configuration
for the random instance $\cal I$ is less than one (for large enough $k$).
So, with probability larger than zero, the random instance $\cal I$ has
no good configuration. Therefore, there exists a fixed instance $\hat{\cal I}$
of the random instance $\cal I$, for which all its configurations are bad.
That is, any feasible solution to instance $\hat{\cal I}$ has size $> c$.
By Lemma \ref{lm - c = Omega(kn^(1/3-epsilon))}, we know that
$c = \Omega(k n^{1/3 - \epsilon})$ for any small constant $\epsilon > 0$.
The lemma follows.
\end{proof}

\subsection{Forcing $z < 1$}
\label{sec - z < 1}

\noindent
{\bf Lemma \ref{lm - z < 1}.} (restated)
\begin{em}
For any positive constants $\beta$ and $\delta$, as long as
$\beta > \delta - 1$,
we will have $z < 1$ for large enough $k$, that is, for any $k \geq k_0$,
where $k_0$ is a constant depending only on $\beta$ and $\delta$.
\end{em}
\begin{proof}[Proof of Lemma \ref{lm - z < 1}]
By (\ref{eqn - definition of a}) and (\ref{eqn - def of c}), we have
\begin{equation}
\label{eqn - a = k^delta}
    a = \frac{c}{k} = k^{\delta}.
\end{equation}
The upper bound of configuration number in Lemma \ref{lm - number of configurations}
is
\[
    (4a+1)^k d^{(4a+1)k} = \exp(k\ln(4a+1) + (4a+1)k\ln d).
\]
By (\ref{eqn - a = k^delta}), we know $k\ln(4a+1) = O(\delta k \ln k)$.
By (\ref{eqn - def of d}) and (\ref{eqn - a = k^delta}), we know
$(4a+1)k\ln d = O(\delta k^{1+\delta} \ln k)$.
So, we have
\[
    (4a+1)^k d^{(4a+1)k} \leq \exp(c' \delta k^{1+\delta} \ln k))
\]
for some constant $c' > 0$.

By our choices of $d$ and $a$ (see (\ref{eqn - def of d}) and
(\ref{eqn - a = k^delta})), we have
\begin{equation}
\label{eqn - ... >= 1/4}
    \left( 1 -\frac{8a}{d} \right)^{4a}
    = \left[\left( 1 -\frac{8a}{d} \right)^{\frac{d}{8a}}\right]^{\frac{32a^2}{d}}
    \geq \left(\frac14\right)^{\frac{32a^2}{d}}
    = \frac14,
\end{equation}
where the inequality is due to that $d \geq 2\cdot 8a$ and hence
$\left( 1 -\frac{8a}{d} \right)^{\frac{d}{8a}} \geq \frac14$,
and the last equality is due to $d = 32a^2$
(This is why we set $d$ as that in (\ref{eqn - def of d})).
Therefore, the probability upper bound in
Lemma \ref{lm - prob that configuration is good} is
\[
    \left[1 - \left( 1 -\frac{8a}{d} \right)^{4a}\right]^{\frac12 h (3k/4)(3k/4 - 1)}
    \stackrel[\text{(\ref{eqn - ... >= 1/4})}]{}{\leq}
        \left(\frac34\right)^{\frac12 h (3k/4)(3k/4 - 1)}
    \stackrel[\text{(\ref{eqn - def of h})}]{}{\leq}
        \exp(-c''k^{2+\beta})
\]
for some constant $c'' > 0$.

Therefore, by (\ref{eqn - def of z}), the definition of $z$, we have
\begin{equation}
\label{eqn - z <= 2^...}
    z \leq \exp(c'\delta k^{1+\delta}\ln k - c''k^{2+\beta}).
\end{equation}

Let us focus on the exponent of the right hand side of (\ref{eqn - z <= 2^...}).
In fact, as long as
\[
%\label{eqn - choose of beta}
    2+\beta > 1+\delta,
\]
we will have
\begin{equation}
\label{eqn - how large k should be}
    c''k^{2+\beta} > c'\delta k^{1+\delta}\ln k
\end{equation}
for large enough $k$.
Consequently, the exponent of the right hand side
of (\ref{eqn - z <= 2^...}) will be negative and we really will have $z < 1$.

One can verify that $k$ can be {\em any} integer larger than
a sufficiently large constant, say $k_0$, depending on $\beta$ and $\delta$.
By the following proof of Lemma \ref{lm - c = Omega(kn^(1/3-epsilon))},
$\beta$ and $\delta$ are two constants depending only on $\epsilon$.
So, $k_0$ is a constant depending only on $\epsilon$, too.
%(For example, if $\epsilon = 0.01$, $k_0$ would be $10^{1000}$.)
%This is our choice of $k$.
\end{proof}

\noindent
{\bf Lemma \ref{lm - c = Omega(kn^(1/3-epsilon))}.} (restated)
\begin{em}
Let $\epsilon$ be any small constant in $(0, 1/3)$. There exist
appropriate positive values of $\delta$ and $\beta$, which only depend on $\epsilon$,
such that $c = \Omega(k n^{1/3 - \epsilon})$.
\end{em}
\begin{proof}[Proof of Lemma \ref{lm - c = Omega(kn^(1/3-epsilon))}]
Recall that $n$ is the number of vertices in graph $G$.
By (\ref{eqn - n = Theta(k^2 dh)}), (\ref{eqn - q = kd}), and (\ref{eqn - def of h}),
we know $n = \Theta(k^{2} d h) = \Theta(k^{2+2\delta+\beta})$.
So, we have
\begin{equation}
\label{eqn - k^delta = Theta(n^(...))}
k^{\delta} = \Theta(n^{\frac{\delta}{2+2\delta+\beta}}).
\end{equation}

Recall from (\ref{eqn - def of c}) that $c = k \cdot k^{\delta}$.
Our goal is to make the exponent
\[
    \frac{\delta}{2+2\delta+\beta}
\]
in (\ref{eqn - k^delta = Theta(n^(...))}) as large as possible to get
a large enough integrality gap. So, $\beta$ should be as small as possible,
meanwhile it should satisfy (\ref{eqn - condition that beta > delta-1}).
Therefore, for any small constant $\epsilon \in (0, 1/3)$, we set
\[
    \beta := \delta-1+\frac{\epsilon}{2}
\]
and choose $\delta$ to be {\em any} constant satisfying
\[
    \delta \geq \frac{1}{3\epsilon}.
\]
Note that there are multiple choices of $\delta$ and $\beta$.

With the values of $\delta$ and $\beta$ chosen as above,
we have
\begin{equation}
\label{eqn - exponent is larger than 1/3-epsilon}
\frac{\delta}{2\delta+\beta+2} > \frac13 - \epsilon.
\end{equation}
By (\ref{eqn - def of c}), (\ref{eqn - k^delta = Theta(n^(...))}),
and (\ref{eqn - exponent is larger than 1/3-epsilon}), we have
\begin{equation}
\label{eqn - lower bound on c}
    c = \Omega(kn^{1/3-\epsilon}).
\end{equation}
This gives the lemma.
\end{proof}

\section{Conclusions}
\label{sec - conclusion and discussion}
In this paper, we prove that two natural linear program relaxations
for {\sf Label $s$-$t$ Cut} ((\ref{LP1 - weak LP for label cut}) and
(\ref{LP2 - strong LP for label cut}) in the paper) have large integrality gaps
$\Omega(m)$ and $\Omega(m^{1/3-\epsilon})$. These are two theoretical lower bound
results.

For the {\sf Label $s$-$t$ Cut} instances we construct
in Section \ref{sec - LP1} and Section \ref{sec - construction of the instance},
it is easy to see that we can make the graphs in the instances directed
by orienting every edge from $s$ to $t$. The analyses of integrality gaps
of (\ref{LP1 - weak LP for label cut}) and (\ref{LP2 - strong LP for label cut})
still go through for the resulting directed graphs. So, the integrality gap
$\Omega(m)$ of (\ref{LP1 - weak LP for label cut}) and the integrality gap
$\Omega(m^{1/3-\epsilon})$ of (\ref{LP2 - strong LP for label cut})
naturally extend to the directed {\sf Label $s$-$t$ Cut} problem.

Until now, we know that {\sf Label $s$-$t$ Cut} has high approximation hardness
factor \cite{ZCTZ11} and its two natural LP-relaxations have large integrality gaps.
A challenging problem for {\sf Label $s$-$t$ Cut}
is to improve its approximation hardness or approximation ratio further.
Either direction seems not easy. Experimental results are also welcome
for {\sf Label $s$-$t$ Cut}.

A closely related challenging problem is to determine the exact complexity
of {\sf Global Label Cut}. Until now, we do not know whether it is in P or
NP-hard.

\section*{Acknowledgements}
Peng Zhang is supported by
the National Natural Science Foundation of China (61672323),
the Natural Science Foundation of Shandong Province (ZR2016AM28),
and the Fundamental Research Funds of Shandong University (2017JC043).

\appendix

\section*{Appendix}

\section{Origins of {\sf Label $s$-$t$ Cut}}
\label{sec - origins of label s-t cut}
The {\sf Label $s$-$t$ Cut} problem came from the work of Jha et al. \cite{JSW02},
Sheyner et al. \cite{SHJ+02}, and Sheyner et al. \cite{SW04}
in computer security, in particular on intrusion detection and on generation
and analysis of the so-called ``attack graphs''.
In this application, an attack graph $G$ has nodes
representing various states, and directed edges representing state
transitions and are labeled by possible ``atomic attacks''.
A pair of special nodes $s$ and $t$ are also given representing
the initial state and the success state (for the intruder).
If the intruder's current state becomes $t$, it means that the intruder
has successfully intruded the system.
The defender's task is to avoid the intrusion by disabling some atomic attacks
of the intruder. To disable an atomic attack incurs some cost
(a unit or a weighted cost). Then the computational task is to find a subset
of atomic attacks of minimum cardinality (or of minimum total weight),
such that the removal of all edges labeled by these atomic attacks disconnects
$s$ and $t$. This is precisely the (directed) {\sf Label $s$-$t$ Cut} problem.

Very interestingly, the {\sf Label $s$-$t$ Cut} problem independently
arose in the research of network survivability by Coudert et al. \cite{CDP+07}.
In a virtual network
(e.g., IP/WDM and MPLS networks \cite{CDP+07}, VPN (virtual private network), etc),
what lie between the network nodes are logical connections, which are realized
via the underlying physical paths consisting of physical links.
In other words, a virtual network is a logical network which
is built on the underlying physical communication network.
See Figure \ref{fig - logical network and physical network} for an example.
(Figure \ref{fig - logical network and physical network} originally appeared
in \cite{CDP+07}.)
If some physical link (i.e., edge in the underlying network) fails,
then all the logical connections (i.e., edges in the virtual network) that use
this physical link fail. We can identify each physical link with a distinct label.
In this way, we get a labeled virtual network, in which a logical connection $e$
has a label $\ell$ if and only if $e$ is realized using the physical link $\ell$.
The key point here is that, different logical connections may share the same label,
and removing one label may destroy several logical connections at the same time.
Then the {\sf Label $s$-$t$ Cut} problem gives tight lower bound
on the number of failures on physical links that can disconnect
a given node pair $(s, t)$ in a virtual network.

\begin{figure}
\begin{center}
\includegraphics*[width=0.7\textwidth]{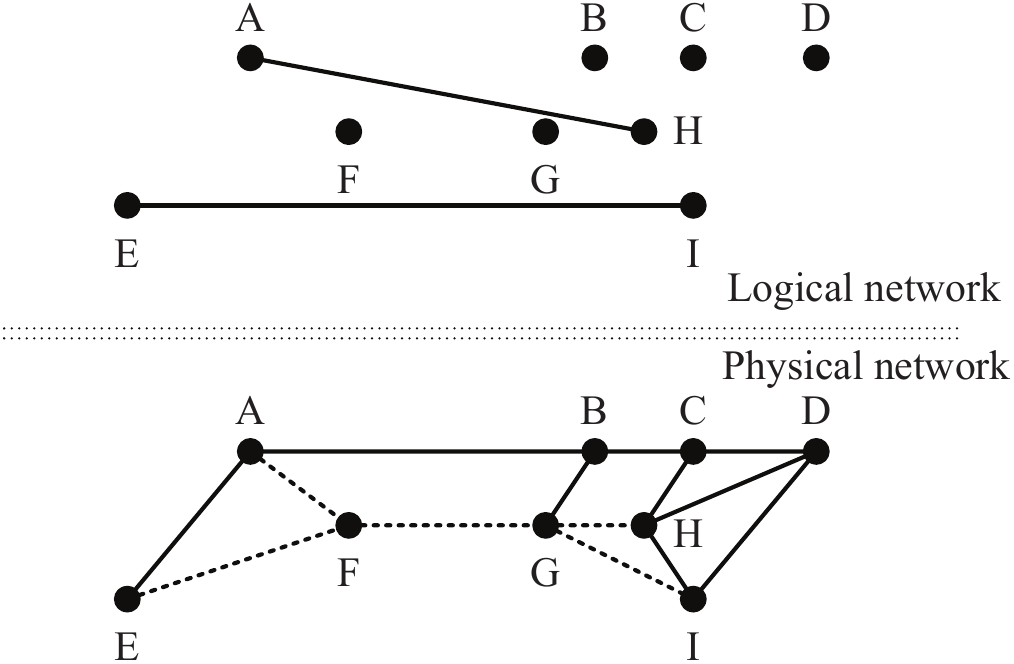}
\end{center}
\caption{A logical network and its underlying physical network.
Only logical connections AH and EI are drawn in the logical network.
AH is realized via the physical path AFGH, while EI is realized via
the physical path EFGI. Note that AH and EI share the same physical link
FG. \label{fig - logical network and physical network}}
\end{figure}

For the application in virtual networks, there is a subtle thing we need
to clarify. Since a logical connection is realized via a physical path,
and a physical path may consist of several physical links, a logical connection
may have more than one labels in general. That is, in the {\sf Label $s$-$t$ Cut}
instance we get, there may be more than one labels on an edge.
However, we can reasonably prescribe that once one of these labels is removed,
the logical connection is destroyed.

%The {\sf Label $s$-$t$ Cut} problem is very natural that it can appear in many applications.
%For example, in large scale networks like social networks, different labels
%on edges can be used to denote various transmission methods of harmful
%information (e.g., rumors or malicious codes) between vertices.
%Then the {\sf Label $s$-$t$ Cut} problem can be used to find the minimum cost way to
%keep an important vertex $t$ away from a contaminated vertex $s$.

\bibliographystyle{plain}
\bibliography{lcbib}

\end{document}